\documentclass{scrartcl}
\usepackage[utf8]{inputenc}
\usepackage[hidelinks,breaklinks]{hyperref}
\usepackage{breakurl}
\usepackage[T1]{fontenc}
\usepackage[english]{babel}
\usepackage{amsthm}
\usepackage{amsmath}
\usepackage{algpseudocode}
\usepackage{algorithm}
\usepackage{bm}
\usepackage{enumitem}
\usepackage{natbib}
\usepackage{mathtools}
\usepackage[svgnames]{xcolor}

\hypersetup{
    colorlinks,
    linkcolor={red!50!black},
    citecolor={blue!50!black},
    urlcolor={blue!30!black}
}

\usepackage[
    type={CC},
    modifier={by},
    version={4.0},
]{doclicense}

\usepackage[bibliography=common]{apxproof}

\usepackage[nameinlink]{cleveref}
\crefname{figure}{Figure}{Figures}
\crefname{example}{Example}{Examples}
\crefname{theorem}{Theorem}{Theorems}
\crefname{lemma}{Lemma}{Lemmas}
\crefname{proposition}{Proposition}{Propositions}
\crefname{claim}{Claim}{Claims}
\crefname{conjecture}{Conjecture}{Conjectures}
\crefname{section}{Section}{Sections}
\crefname{definition}{Definition}{Definitions}

\AddToHook{env/proposition/begin}{\crefalias{theorem}{proposition}}
\AddToHook{env/lemma/begin}{\crefalias{theorem}{lemma}}
\AddToHook{env/example/begin}{\crefalias{theorem}{example}}
\AddToHook{env/definition/begin}{\crefalias{theorem}{definition}}
\AddToHook{env/corollary/begin}{\crefalias{theorem}{corollary}}
\AddToHook{env/claim/begin}{\crefalias{theorem}{claim}}
\AddToHook{env/remark/begin}{\crefalias{theorem}{remark}}

\newtheoremrep{theorem}{Theorem}[section]
\newtheoremrep{conjecture}[theorem]{Conjecture}
\newtheoremrep{lemma}[theorem]{Lemma}
\newtheoremrep{proposition}[theorem]{Proposition}
\newtheoremrep{definition}[theorem]{Definition}
\newtheoremrep{remark}[theorem]{Remark}
\newtheoremrep{claim}[theorem]{Claim}
\newtheoremrep{corollary}[theorem]{Corollary}
\newtheoremrep{example}[theorem]{Example}

\usepackage{xspace}
\usepackage{amsmath}
\usepackage{amsfonts}

\usepackage{tikz}
\usetikzlibrary{positioning, arrows.meta}
\usetikzlibrary{decorations,decorations.pathmorphing}

\colorlet{colscc}{brown}

\newcommand{\myparagraph}[1]{\paragraph{#1.}}

\title{Cutwidth Bounds via Vertex Partitions}
\author{Antoine Amarilli$^1$, Benoît Groz$^2$\\
$^1$: Univ.\ Lille, Inria, CNRS, Centralle Lille, UMR 9189 CRIStAL\\
$^2$: Paris-Saclay University, CNRS, LISN\\
\texttt{antoine.a.amarilli@inria.fr}, \texttt{benoit.groz@lisn.upsaclay.fr}
}
\date{}

\begin{document}

\maketitle

\begin{abstract}
  We study the cutwidth measure on graphs and ways to bound the cutwidth of a
  graph by partitioning its vertices. We consider bounds expressed as a function
  of two quantities: on the one hand, the maximal cutwidth $y$ of the subgraphs
  induced by the classes of the partition, and on the other hand, the cutwidth
  $x$ of the quotient multigraph obtained by merging each class to a single
  vertex. We consider in particular the decomposition of directed graphs into
  strongly connected components (SCCs): in this case, $y$ is the maximal
  cutwidth of an SCC, and $x$ is the cutwidth of the directed acyclic
  condensation multigraph.

  We show that the cutwidth of a graph is always in $O(x+y)$, specifically it
  can be upper bounded by $1.5x + y$. We also show a lower bound 
  justifying that the constant $1.5$ cannot be improved in general
\end{abstract}

\section{Introduction}

The measure of \emph{cutwidth} (see \cite{cutwidth}) is a parameter on undirected
graphs which intuitively measures the best performance of a linear ordering of
the vertices in terms of how many edges are ``cut''. Cutwidth is related to
\emph{pathwidth}: it is always as least as high as the pathwidth, and the two
parameters coincide up to constant factors when the maximal degree of the graph
is bounded by a constant. One point of cutwidth, like pathwidth and treewidth,
is that bounding it can ensure the tractability of some problems on graphs: when
the input graph is assumed to have constant cutwidth, then we can design
dynamic programming algorithms on the input graph.

In this paper, we study how to show upper bounds on the treewidth of graphs via
the approach of partitioning its vertices into classes, and studying on the one
hand the cutwidth of each class, and on the other hand the cutwidth of the
\emph{quotient} obtained by merging each class to a single vertex. As the
quotient may in general feature parallel edges, and cutwidth accounts for these
parallel edges, then we define the quotient as a multigraph, and work with
multigraphs throughout the paper. Let us assume that we can show an upper bound
of $x$ on the cutwidth of the quotient multigraph, and that we can show an
upper bound of $y$ on the maximal cutwidth of the classes of the partition. We would
then want to show an upper bound on the cutwidth of the original graph, as a
function of $x$ and $y$. Intuitively, we want to do so by ordering the classes
according to the order that achieves cutwidth $x$ on the quotient multigraph,
and then ordering each class according to the order that achieves cutwidth at
most $y$ within each class.

This technique of bounding cutwidth via vertex partitions has been studied
earlier. \cite{barth1995bandwidth} have studied how to bound the
bandwidth or cutwidth of a graph as a function of the same measure on the
quotient graph and on the classes of the partition; however, their bounds depend on
the degree of the graph. More recently, \cite{FeldmannM23}
have applied the method in the specific case where the input graph is directed,
and the classes of the partition are the strongly connected components (SCCs):
the quotient multigraph is then called the \emph{condensation multigraph}. Their
work \cite[Lemma~2.2]{FeldmannM23} claims an upper bound of $x+y$ on the
cutwidth of the graph $G$, again as a function of the (undirected) cutwidth $x$
of the condensation multigraph of $G$ and of the maximal (undirected) cutwidth
$y$ of an SCC of~$G$. This standalone lemma of~\cite{FeldmannM23} is of independent interest, and we
recently applied it to our study of the edge-minimum walk of modular length
problem in~\cite{amarilli2025edge}.

In this paper, we specifically study this question of bounding the cutwidth of a
graph using vertex partitions, namely, as a function of the cutwidth $x$ of the
quotient multigraph and of the maximal cutwidth $y$ of a class. 
We show that the proof of \cite[Lemma~2.2]{FeldmannM23} contains an oversight:
as we explain, their technique shows an upper bound of $2x+y$ and not
$x+y$. Intuitively, the issue is that the order defined on the initial
graph~$G$, during its enumeration of the vertices of a class $C$, may cut some
edges connecting $C$ to both preceding classes and succeeding classes in the
ordering on the quotient multigraph: this means that we may cut more edges than
the maximal number $x$ that we achieve on the quotient multigraph when we
enumerate all vertices of~$C$ at once.

This bounds of $2x+y$ leaves open the question of whether we can improve it, in
particular to reach the bound of $x+y$ claimed in \cite{FeldmannM23}. Our main
technical contribution is to show (in \cref{thm:upper}) a bound of $1.5x+y$.
This bound is obtained via a more careful choice of ordering on the initial
graph $G$: we still order the classes of the partition in the order that
achieves the cutwidth of~$x$ on the quotient multigraph, but for each class
of~$G$ we choose among a minimum-cutwidth ordering (achieving cutwidth at most
$y$) and its reverse (achieving the same cutwidth), intuitively depending on
which one least worsens the cutwidth bound on the quotient multigraph.

Our result achieves a cutwidth bound of $1.5x+y$ and not $x+y$, but as we show,
this is unavoidable (and the bound of $x+y$
claimed in \cite[Lemma~2.2]{FeldmannM23} is in fact not correct).
Specifically, we show in \cref{prp:lower-bound-condensation} that, for
arbitrarily large $x$ and $y$, there are graphs whose cutwidth is $1.5x + y$:
this implies that the factor~$1.5$ cannot be improved in general.

\paragraph*{Paper structure.}
We give preliminaries in \cref{sec:prelim}. We show our upper bound of
$1.5x+y$ (\cref{thm:upper}) in \cref{sec:upper}, and show our lower bound 
(\cref{prp:lower-bound-condensation})
in \cref{sec:lower}. We conclude in \cref{sec:conc}.

\paragraph*{Acknowledgements.}
We are grateful to the authors of \cite{FeldmannM23} for confirming the
issue with \cite[Lemma~2.2]{FeldmannM23}.

\section{Preliminaries}
\label{sec:prelim}
\nosectionappendix

We work with undirected and directed graphs and also with
undirected and directed \emph{multigraphs}, in which edges have multiplicities.
We sometimes talk of a \emph{simple} graph to mean a graph
which is not a multigraph, i.e., the special case of a multigraph where no edge
has multiplicity greater than~$1$.
Formally, we represent multigraphs using \emph{multisets} of edges, where a \emph{multiset}
is a set in which some elements can be repeated (or, equivalently, have a
\emph{multiplicity}, which is a positive integer). 
We write $|S|$ to denote the cardinality of a multiset $S$ (summing over all multiplicities).
We use double curly braces for multisets and 
for multiset comprehension notation, e.g., for $S$ a multiset, if we write $S'
\coloneq \{\{(x,
y) \mid x, y \in S\}\}$, then we mean that $S'$ is the multiset of ordered pairs of elements
of~$S$, taking each element of $S$ with its multiplicity (so that 
$|S'| = (|S|)^2$).
For $S_1$ and $S_2$ two multisets, we write $S_1 \subseteq S_2$ to mean that the
elements of $S_1$ are a subset of that of $S_2$ and the multiplicity of each
element of~$S_1$ is no greater than its multiplicity in~$S_2$.

An \emph{undirected multigraph} $(V, E)$
consists of a set $V$ of vertices and a multiset $E$ of undirected edges
(i.e., pairs of vertices of~$V$), and a \emph{directed multigraph} $(V, E)$ consists of a set $V$ of
vertices and a multiset $E$ of directed edges of the form $(u, v)$ with $u \neq
v$ for $u,v \in V$. Note that we never allow self-loops in graphs.
Given a directed
multigraph $G = (V, E)$, the \emph{underlying undirected multigraph} $G' =
(V, E')$ of~$G$ is obtained
by forgetting the orientation of edges, i.e., we set $E' \coloneq \{\{\{u, v\} \mid (u, v) \in
E\}\}$. Note that $G'$ may be a multigraph even when $G$ is a simple graph,
e.g., if $E = \{(1,2), (2,1)\}$ then $E' = \{\{\{1, 2\}, \{1, 2\}\}\}$.

\myparagraph{Cutwidth}
We consider the measure of \emph{cutwidth}, defined on undirected 
graphs and multigraphs. We will also consider the cutwidth of directed graphs
and multigraphs, but we always define the cutwidth of such graphs to be that of their underlying undirected
multigraphs (i.e., we never consider definitions of cutwidth for directed
graphs, such as the ones of \cite{chudnovsky2012tournament}).
Given a undirected (multi)graph $G = (V, E)$,
a \emph{cut} $V_-, V_+$ of~$V$
is a partition of $V$ into two sets $V_-$ and $V_+$, i.e., $V = V_- \sqcup V_+$
where $\sqcup$ denotes disjoint union.
We say that an edge $e \in E$ \emph{crosses the cut} $V_-,V_+$  if it
intersects both sides, namely, $e \cap V_- \neq \emptyset$ and $e \cap V_+ \neq
\emptyset$.
Then the
\emph{cutwidth} of the cut $V_-,V_+$ is the number of edges that cross the cut, counted
together with their multiplicity, i.e., it is 
$|\{\{e \in E \mid e \cap V_-
\neq \emptyset \text{ and } e \cap V_+ \neq \emptyset\}\}|$.
We then define an \emph{ordering} of~$G$ as a strict total order $<$ on~$V$,
and say that a cut $V_-,
V_+$ of~$V$
\emph{respects}~$<$ 
if we have 
$v_- < v_+$ for each $(v_-, v_+) \in V_- \times V_+$. 
The \emph{cutwidth} of the ordering~$<$ is then the maximum cutwidth of
a cut that respects~$<$. Note that, by symmetry, the cutwidth of an ordering
$<$ is
always the same as that of the reverse ordering $>$ defined by 
$x > y$ iff $y < x$.
Finally, the \emph{cutwidth} of~$G$ is defined as the minimum cutwidth
of an ordering of~$G$.

\myparagraph{Subdivisions}
Letting $G = (V, E)$ be an undirected multigraph, and $e= \{u,v\}$ be (one occurrence of) an
edge, a \emph{subdivision of~$G$ on~$e$} is a multigraph obtained by
removing (one occurrence of) $e$ from~$E$  and replacing it by a path of length 2
that goes via a fresh intermediary vertex. Formally, it is an undirected
multigraph $(V', E')$ obtained by
letting $V' \coloneq V \sqcup \{w\}$ where $w$ is a fresh vertex, and letting
$E'$ be constructed from $E$ by removing (one occurrence of) $e$ and adding the
edges $\{u, w\}$ and $\{w, u\}$. Note that, if we take a multigraph and
perform subdivision on every edge occurrence, then the result is
always a simple graph.
More generally, we define a \emph{multiedge-subdivision of~$G$} as a graph $G' = (V', E')$
as above, but where we consider any number $m$ of occurrences of $e$, with
$m$ is at least~$1$ and at most the number of occurrences of~$e$ in~$E$,
and replace them with $m$ occurrences of the edges $\{u, w\}$ and $\{w, u\}$
(all sharing the same intermediate vertex~$w$). Note that for $m=1$ we recover
the definition of subdivisions given earlier.

We will need the following result about cutwidth and (multiedge-) subdivisions.
In the case of subdivisions, the results appears to be folklore and appears for
instance in~\cite{MakedonS89}. We provide a self-contained proof in
Appendix~\ref{apx:subdiv}
to show that the result also holds for multiedge-subdivisions.

\begin{toappendix}
  \section{Proof of \cref{prp:subdiv}}
  \label{apx:subdiv}
\end{toappendix}

\begin{propositionrep}
  \label{prp:subdiv}
  Let $G = (V,E)$ be a multigraph, and let $G'$ be obtained from $G$ by a multiedge-subdivision.
  Then the cutwidth of~$G'$ is equal to the cutwidth of~$G$. 
\end{propositionrep}

\begin{proof}
  Let $k$ be the cutwidth of~$G$, and let $k'$ be that of~$G'$.
  Let $w$ be the fresh vertex introduced by the multiedge-subdivision, $e = \{u,
  v\}$ the (multi)edge to which we apply the subdivision, $M$ the multiplicity
  of~$e$ in~$G$, and $1\leq m \leq M$ the number of occurrences of $e$ that we
  subdivide: we replace them in~$G'$ by $M$ occurrences of $e_u \coloneq \{u,w\}$ and $e_v \coloneq \{v,w\}$.
  To prove that $k = k'$,
    we will first show that $k'\leq k$ by inserting $w$ at an arbitrary position
    between $u$ and $v$ in an ordering $<$ of $G$ having cutwidth $k$: we will
    show that the resulting ordering $<'$ of $G'$ has cutwidth $\leq k$.
   Reciprocally, given an ordering of $G'$, we will consider the ordering
   induced on~$G$ by omitting $w$, and show that $k\leq k'$.

  \bigskip
  
  We first show the first direction ($k'\leq k$). Given an ordering $<$ of~$V$ which achieves cutwidth~$k$
  for~$G$, we assume up to reversing~$<$ that $u$ occurs before~$v$
  in~$<$. We then build an ordering ${<'}$ of~$V'$ by inserting the fresh
  vertex~$w$ at an arbitrary position between $u$ and~$v$. For any cut
  $V_-', V_+'$ of~$<'$,
  considering the corresponding cut $V_-, V_+$ of~$<$, every edge crossing the cut $V_-', V_+'$ in~$G'$ except the
  two fresh edges $e_u$ and $e_v$ also crosses the cut $V_-, V_+$. For the two fresh multiedges, we know that:
  \begin{itemize}
    \item $e_u$ crosses the cut precisely when $u \in V_-'$ but $w \in V_+'$, and in this
      case the edge~$e$ crosses the cut $V_-, V_+$ because $v$ comes after $w$ in
      $<'$ so $v \in V_+$ and of course $u \in V_-$;
    \item $e_v$ crosses the cut precisely when $w \in V_-'$ but $v \in V_+'$, and in this
      case also the edge~$e$ crosses the cut $V_-, V_+$ because $u$ comes before~$w$
      in~$<'$ so $u \in V_-$ and $v \in V_+$.
  \end{itemize}
  It follows that only one of the multiedges $e_u$ or $e_v$ crosses the cut, and
  then the $m$ occurrences of that multiedge which are crossed in~$<'$
  correspond to the $m$ additional occurrences of $e$ that are cut in~$<$. Thus, the cutwidth of $V_-', V_+'$ in $G'$ is no more than the cutwidth of $V_-, V_+$ in $G$, so the
  cutwidth of~${<'}$ is at most $k$, which establishes that
  $k' \leq k$.

  \bigskip

  We then show the second direction ($k\leq k'$). Given an ordering ${<'}$ of~$V'$ which achieves
  cutwidth~$k'$ for~$V'$, we assume again up to reversing ${<'}$ that $u$ occurs
  before~$v$. We build the ordering $<$ of~$V$ simply by removing the
  fresh vertex~$w$. Consider any cut $V_-, V_+$ of~$<$, and consider some cut
  $V_-', V_+'$ of~${<'}$ which is equal to $V_-, V_+$ 
  up to the fresh vertex, i.e., we have
  $V_- \subseteq V_-' \subseteq V_- \sqcup \{w\}$ and likewise
  $V_+ \subseteq V_+' \subseteq V_+ \sqcup \{w\}$. For any edge of~$G$ which is
  cut by~$V_-, V_+$ and is different from~$e$, we know that the same edge
  in~$G'$ is also cut in~$V_-', V_+'$. Now, the edge~$e$ crosses the cut precisely
  when $u \in V_-$ but $v \in V_+$. In this case, we have $u \in V_-'$ and $v
  \in V_+'$, so it must be the case that one of the edges $e_u$ and $e_v$
  crosses the cut $V_-', V_+'$, which compensates the $m$ additional occurrences
  of~$e$ in~$G$ that are cut by $V_-,V_+$. Hence, the cutwidth of~$V_-, V_+$ is no more than that
  of~$V_-', V_+'$, so that $k \leq k'$. This shows that $k = k'$ and concludes
  the proof.
\end{proof}

\myparagraph{Vertex partitionings}
We will bound the cutwidth of undirected (multi)graphs $G$ using \emph{vertex
partitionings}, which we now define. A \emph{vertex partitioning} of a set~$V$
of vertices is simply a partition $P = \{C_1, \ldots, C_k\}$ of~$V$: each
\emph{class} $C_i$ for $1 \leq i \leq k$ is a nonempty subset of~$V$, the classes are pairwise disjoint, and 
their union
is~$V$. Given a vertex partitioning of the vertices of an undirected (multi)graph~$G = (V, E)$, we can distinguish two kinds of
edges of~$G$: the \emph{internal edges}, where both endpoints belong to the same
class; and the \emph{external edges}, which connect endpoints belonging to
different classes.

A \emph{subgraph} of a (multi)graph~$G = (V, E)$ is simply a (multi)graph $(V, E')$ with $E'
\subseteq E$, and the subgraph \emph{induced} by a subset $C \subseteq V$  of
vertices is the subgraph $(V, E_{|C})$ where we keep the internal edges of
$E$ for which both endpoints are in~$C$ (with their multiplicities), formally,
$E_{|C} \coloneq \{\{e \in E \mid
e \subseteq C\}\}$. We will often abuse notation and identify the classes of~$P$
with the subgraphs of~$G$ that they induce.

From the undirected (multi)graph~$G$ and the vertex partition $P =
\{C_1, \ldots, C_k\}$, in addition to
the subgraphs induced by the $C_i$ for $1 \leq i \leq k$, we
will also consider the \emph{quotient multigraph} where we
intuitively merge each class of the partition into a single vertex, and only
keep the external edges between these vertices. Formally,
let $\phi_P \colon V \to P$ be the function that maps each vertex
of~$V$ to its class in~$P$. The \emph{quotient
multigraph} $G/P$ of~$G$
under~$P$ is then the undirected multigraph with vertex set $P$, and
with the multiset of edges $\{\{\{\phi_P(u), \phi_P(v)\} \mid \{u, v\} \in
E \text{~s.t.~} \phi_P(u) \neq \phi_P(v)\}\}$. Note that the quotient
multigraph is in general a multigraph even if $G$ is a simple graph, e.g., if
$G$ contains the edges $\{1, 3\}$ and $\{2, 3\}$ and $P$ contains the classes
$\{1, 2\}$ and $\{3\}$.

We also define quotient graphs over directed (multi)graphs in the expected way,
in particular for the case of SCCs that we define next. Formally, for~$G = (V,
E)$ a directed (multi)graph, for $P$ a vertex partitioning of~$V$,
and for $\phi_P \colon V \to P$ the function mapping the vertices of~$V$ to
their classes in~$P$,
the \emph{quotient multigraph} of~$G$ by~$P$
is the directed multigraph $G/P$ whose vertex set is~$P$ and whose multiset of edges is
$\{\{(\phi_P(u), \phi_P(v)) \mid (u, v) \in E\}\}$.

\myparagraph{SCCs}
One special case of vertex partitioning which we will study
is the well-known decomposition of
directed (multi)graphs into \emph{Strongly Connected
Components} (SCCs). Specifically, we say that two vertices $u$ and $v$ of a directed
(multi)graph~$G = (V, E)$ are in the
same SCC if there is a directed path from~$u$ to~$v$ in~$G$ and a directed
path from~$v$ to~$u$ in~$G$: this relation is clearly an equivalence relation,
so it defines a vertex partitioning $P$. Note that, for each class $C$ of~$P$,
the subgraphs of~$G$ induced by the classes of~$P$ are exactly the strongly
connected components of~$G$ in the usual sense.
Further, the quotient multigraph $G/P$ is then called the \emph{condensation multigraph}
$G'$ of~$G$, i.e., $G'$ is the directed acyclic multigraph obtained by condensing the SCCs
of~$G$.
Note that, in most of the works that
consider the decomposition of a directed graph into SCCs,
the condensation graph is defined as a directed acyclic simple graph over
SCCs: but in this
note we see $G'$ as a (directed acyclic) \emph{multigraph}
(like in~\cite{FeldmannM23}),
with the multiplicity of every edge from $C$
to~$C'$ in~$G'$ intuitively denoting how many edges of~$G$ go from a vertex
in~$C$ to a vertex in~$C'$. This distinction will be important when studying the
cutwidth of the condensation multigraph.

\section{Upper bound}
\label{sec:upper}

In this section, we show the following general result on vertex partitionings of
undirected multigraphs:

\begin{theorem}
  \label{thm:upper}
Let $G = (V,E)$ be an undirected multigraph,
  $P$ a partition of $G$, and $G/P$ be the
quotient multigraph of~$G$ under~$P$.
  If the cutwidth of $G/P$ is $x$ and the cutwidth of every class of $P$ is at most $y$, then the cutwidth of $G$ is at most $1.5x + y$.
\end{theorem}

This result holds for arbitrary vertex partitionings. If we instantiate it to
the specific case of vertex partitioning defined by SCCs on
directed graphs, we obtain the following corollary:

\begin{corollary}
Let $G$ be a directed multigraph
  and $G'$ be its condensation multigraph.
  If the cutwidth of $G'$ is $x$ and the cutwidth of every SCC of $G$ is at most $y$, then the cutwidth of $G$ is at most $1.5x + y$.
\end{corollary}

Remember that, here, the cutwidths are always taken on the underlying undirected multigraphs. The statement of the corollary is similar to~\cite[Lemma~2.2]{FeldmannM23}, except that the
bounds are worse ($1.5x + y$ instead of $x+y$). 
In the rest of this section, we prove Theorem~\ref{thm:upper}. We first show a
weaker bound of $2x+y$ with a simple approach that uses the ordering defined
in~\cite{FeldmannM23}.
We next show how a careful choice of ordering lowers the bound to $1.5x+y$.
We will show in the next section that this bound cannot be improved in
general.

\myparagraph{Simple bound in $\bm{2x+y}$}
Let $G = (V,E)$ be an undirected multigraph and $P$ be a vertex partitioning of $G$.
Remember that we identify each class $C$ of~$P$ with the subgraph of~$G$ that it induces.
We will define an ordering of~$G$, and we will require that it is defined by
first choosing an ordering on the classes of~$P$, and then ordering the vertices
inside each class~$C$ according to some ordering on~$C$. Following the terminology
of~\cite[Section~2.3]{barth1995bandwidth}, we say that such an ordering of~$G$
is \emph{compatible} with~$P$:
we will only consider such orderings in the proof. 

For each class $C\in P$, we let $<_C$ be an ordering on $C$ that achieves a
cutwidth of at most~$y$ on~$C$.
Further, let $<_{G/P}$ be an ordering on the quotient multigraph $G/P$ that
achieves the optimal cutwidth of~$x$ on $G/P$. We now define a compatible ordering
$<$ of $G$ like in~\cite{FeldmannM23}: we first order
vertices following the order of classes given by $<_{G/P}$, and then inside each class~$C$ we order the vertices
according to the ordering $<_C$.
Formally, for $u, v \in V$, letting $C_i$ and $C_j$ be the respective classes of
$u$ and~$v$, we set $u < v$ if we have $C_i <_{G/P} C_j$ or if we have $C_i=C_j$ and $u
<_C v$ for $C_i = C = C_j$.

Let $C$ be a class, and let $V_-'$ (resp.\ $V_+'$) consist of all vertices of
classes before $C$ (resp.\ after $C$) in $<_{G/P}$. 
Any cut $C_-, C_+$ of $C$ induces a cut $V_-'\sqcup C_-, C_+\sqcup V_+'$ on $G$ that
respects the ordering~$<$.
Let $\Xi$ be the multiset of edges of~$E$ that cross the cut. 
We can partition $\Xi$ into the following disjoint sets:
\[
  \Xi = \Xi_{C, C} \sqcup \Xi_{V_-', C} \sqcup \Xi_{C, V_+'} \sqcup \Xi_{V_-', V_+'},
\]
where $\Xi_{X,Y}$ is the set of edges crossing the cut with one extremity in $X$ and the other in $Y$.
For instance, $\Xi_{C, C}$ is the subset of the edges internal to $C$ which cross the cut, whereas $\Xi_{V_-', C}$ are the edges with one  extremity in $V_-'$ and the other extremity in $C$ which cross the cut (and therefore $\Xi_{V_-', C} = \Xi_{V_-', C_+}$).

The number of edges $|\Xi_{C, C}|$ is bounded by $y$ because the restriction of $<$
to~$C$ achieves a cutwidth of at most~$y$, and the quantities $|\Xi_{V_-', C}| + |\Xi_{V_-', V_+'}|$ 
and $|\Xi_{C, V_+'}| + |\Xi_{V_-', V_+'}|$ 
are each bounded by $x$
by considering the order $<_{G/P}$ on the quotient multigraph.
Hence, the number of edges crossing the cut is at most  
$|\Xi_{V_-', C}| + |\Xi_{C, V_+'}| + |\Xi_{V_-', V_+'}| + |\Xi_{C, C}| \leq 2x + y$.
We have thus proved that ordering $<$ has cutwidth at most $2x + y$.

\myparagraph{More elaborate bound in $\bm{1.5x+y}$}
To show the claimed bound of $1.5x+y$, we now define the ordering $<$ more
carefully. We will still order vertices first according to the order $<_{G/P}$,
so that in particular the order is still compatible:
but then within each class $C$ we will choose between the ordering $<_C$ and its reverse $>_C$.

Let $C$ be a class, and let $V_-'$ (resp.\ $V_+'$) consist of all vertices of
classes before $C$ (resp.\ after $C$) in $<_{G/P}$. 
As we explained previously, every cut $C_-, C_+$ of $C$ according to $<_C$
induces a cut $V_-'\sqcup C_-, C_+\sqcup V_+'$ on $G$.
The same split of $C$ into $C_+, C_-$ is also a cut of $C$ according to $>_C$
where the vertices of $C_+$ now come before those of $C_-$ and this cut
according to $>_C$ induces a cut $V_-'\sqcup C_+, C_-\sqcup V_+'$ on $G$.

We can partition the edges of $G$ as follows, where $E_{X,Y}$ denotes the edges having an extremity in $X$ and the other in $Y$. The partitioning is similar to the one above, except that we partition all the edges and not only those crossing the cut of $C$ according to $<_C$, and we distinguish more edge types:
\begin{equation}
  \label{eq:decomp}
  E = E_{V_-', V_-'}
\sqcup E_{V_+', V_+'}
\sqcup E_{C, C}
\sqcup E_{V_-', C_-}
\sqcup E_{C_+, V_+'}
\sqcup E_{V_-', C_+}
\sqcup E_{C_-, V_+'}
\sqcup E_{V_-', V_+'}
\end{equation}

\begin{itemize}
\item The edges from $E_{V_-', V_-'} \sqcup E_{V_+', V_+'}$ (internal to other
  classes than~$C$, or connecting classes located on the same side of the cut):
    these edges will never cross the cut (no matter whether we order vertices of
    the various classes $C'$ according to $<_{C'}$ or according to $>_{C'}$).

\item The edges from $E_{C,C}$ are internal to $C$: at most~$y$ of these cross the cut, because we will order the vertices of~$C$ either according to~$<_C$ or according to~$>_C$ (and both achieve the same cutwidth).

\item The edges from $E_{V_-', C_+} \sqcup E_{C_-, V_+'}$
cross the cut $V_-' \sqcup C_-, C_+ \sqcup V_+'$ according to $<_C$
(but not the cut $V_-'\sqcup C_+, C_-\sqcup V_+'$ according to $>_C$).

\item The edges from $E_{V_-', C_-} \sqcup E_{C_+, V_+'}$ 
cross the cut $V_-'\sqcup C_+, C_-\sqcup V_+'$ according to $>_C$
(but not the cut $V_-' \sqcup C_-, C_+ \sqcup V_+'$ according to $<_C$).

\item The edges from $E_{V_-', V_+'}$ always cross the cut (no matter whether we
  order the vertices of the classes $C'$ according to $<_{C'}$ or according to
    $>_{C'}$).
\end{itemize}

Figure~\ref{fig:cw-quotient-edges} illustrates all the edge types which may cross the cut.
We make the following claim, which implies that one of $<_C$ and its reverse will achieve the cutwidth bound:
\begin{claim}
  \label{clm:choice}
  Let $<_{G/P}$ be an ordering on $G/P$ that achieves the optimal cutwidth $x$
  on $G/P$. Let $V_-'$ (resp.\ $V_+'$) consist of all vertices of classes before
  (resp.\ after) $C$ in $<_{G/P}$.
  For every class $C$ and ordering $<_C$ on $C$, at least one of the following claims is true:
  
  \begin{itemize}
    \item For every cut $C_-, C_+$ of~$C$ according to~$<_C$, we have 
$|E_{V_-', C_+}| + |E_{C_-, V_+'}| + |E_{V_-', V_+'}|\leq 1.5x$
    \item For every cut $C_-, C_+$ of~$C$ according to~$<_C$, we have 
$|E_{V_-', C_-}| + |E_{C_+, V_+'}| + |E_{V_-', V_+'}|\leq 1.5x$
  \end{itemize}
\end{claim}

This claim suffices to conclude the proof of
\cref{thm:upper}. Indeed, the first (resp.\ second) case in the claim bounds the
number of external edges crossed in cuts compatible with $<_C$ (resp., with
$>_C$). We can therefore use the claim on each class $C$ to choose to order its
vertices either according to~$<_C$ or according to~$>_C$.
Overall, this gives a compatible order $<$, of which we can show that it has
cutwidth at most $1.5x + y$. Indeed, let us consider a cut $V_-,V_+$. 
If $V_+$ is empty no edge can cross the cut, so we may assume that $V^+$ is
non-empty. Let $C$ be the first class in $<_{G/P}$ which contains a vertex from
$V^+$. Then $V_-, V_+$ is the cut induced by $C_-, C_+$ where $C_- = V_-\cap C$
and $C_+ = V_+\cap C$. We can then bound the number of external edges crossed in
the cut $V_-,V_+$ by $1.5x$ from the claim statement (because we chose to order
$C$ according to~$<_C$ or to~$>_C$ depending on which of the two statements of
the claim is true), and the number of internal
edges crossed in the cut is bounded by the number of internal edges in~$C$
crossed in the cut, which is at most~$y$ because the order on~$C$ is
either~$<_C$ or~$>_C$.
This allows us to conclude the proof of \cref{thm:upper}.

Hence, all that remains is to show Claim~\ref{clm:choice}. The intuition is that
we consider the worst cut according to the first equation: either the cut in
question satisfies the bound of $1.5x$ and we use the first case, or it does
not. In this second case, partitioning the external edges involving vertices
of~$C$ in 4 types according to that cut (like the four last sets of
Equation~\ref{eq:decomp}), we show how the violation of the first case implies
that the second case is satisfied by any cut of~$C$. The formal proof follows:

\begin{proof}[Proof of Claim~\ref{clm:choice}]
  Let us first observe that, given that the ordering $<_{G/P}$ achieves a
  cutwidth bound of at most~$x$ on the quotient multigraph~$G/P$, then by
  considering the cut right before~$C$ (i.e., the cut $V_-', C \sqcup V_+'$)
  we know that $|E_{V_-', C}| + |E_{V_-', V_+'}|\leq x$, where $E_{V_-',C}$
  denotes the external edges involving a vertex of~$C$ and a vertex of a class before~$C$ in~$<_{G/P}$.
  Consequently,
  for any cut $C_-,C_+$ of~$C$ according to~$<_C$, 
  given that $E_{V_-',C} = E_{V_-',C_-} \sqcup E_{V_-',C_+}$,
  we have:
\begin{equation}\label{eq:cutleft}
|E_{V_-', C_-}| + |E_{V_-', C_+}| + |E_{V_-', V_+'}|\leq x
\end{equation}
and similarly, considering the cut right after $C$, we obtain:
\begin{equation}\label{eq:cutright}
|E_{C_-, V_+'}| + |E_{C_+, V_+'}| + |E_{V_-', V_+'}|\leq x
\end{equation}
These 2 cuts are depicted by dotted vertical lines on the left of Figure~\ref{fig:cw-quotient-edges}.

We now fix a cut $C_-,C_+$ of~$C$ according to~$<_C$ with $C_+$ nonempty
  that maximizes the number of external edges
  crossing the cut 
  when ordering $C$ according to~$<_C$,
  i.e., which maximizes the value $n \coloneq |E_{V_-', C_+}| + |E_{C_-, V_+'}|
  + |E_{V_-', V_+'}|$ in the left-hand-side of the inequality in the first case
  of Claim~\ref{clm:choice}.
  There are two cases: either we have $n \leq 1.5x$, or we have $n > 1.5x$.
  In the first case, we can conclude immediately that the first case of
  Claim~\ref{clm:choice} holds, because our choice for the cut $C_-, C+$ makes
  $n$ an upper bound on the left-hand side of the first inequality. Hence, in the rest of the proof, we focus on the
  second case, where we have:
  \begin{equation}\label{eq:case-largecut}
|E_{V_-', C_+}| + |E_{C_-, V_+'}| + |E_{V_-', V_+'}| > 1.5x 
\end{equation}
Now, combining Equations~$\eqref{eq:cutleft}+\eqref{eq:cutright}-\eqref{eq:case-largecut}$, we deduce:
\begin{equation}\label{eq:cutothers}
|E_{V_-', C_-}| + |E_{C_+, V_+'}| < 0.5x
\end{equation}
which implies in particular the following:
  \begin{equation}
    \label{eq:mm}
    |E_{V_-', C_-}| < 0.5x
  \end{equation}
  \begin{equation}
    \label{eq:pp}
  |E_{C_+, V_+'}| < 0.5x
  \end{equation}

We will use these properties of the fixed cut $C_-,C_+$ 
  to show that the second case of Claim~\ref{clm:choice} holds.  Let us therefore show that,
  for every cut $C_-', C_+'$ of~$C$ according to~$<_C$, we can bound the
  cutwidth of the cut $C_+',C_-'$ for~$>_C$, i.e., we have the bound 
$|E_{V_-', C_-'}| + |E_{C_+', V_+'}| + |E_{V_-', V_+'}|\leq 1.5x$, which is the
  inequality in the second case of Claim~\ref{clm:choice}.
To do so, we will use partition the external edges involving~$C$ in 4 types according to the fixed cut $C_-,C_+$ chosen earlier. 

  We distinguish two cases:
  \begin{itemize}
    \item The cut $C_-',C_+'$ is ``before'' $C_-, C_+$, i.e., we have $C_-'
      \subseteq C_-$ and $C_+' \supseteq C_+$. In this case, we have $E_{V_-', C_-'} \subseteq E_{V_-',
      C_-}$, and we always have $E_{C_+', V_+'} \subseteq E_{C_+, V_+'} \sqcup
      E_{C_-, V_+'}$. Thus, we have:
      \[
        |E_{V_-', C_-'}| + |E_{C_+', V_+'}| + |E_{V_-', V_+'}| \leq 
        |E_{V_-', C_-}| + |E_{C_+, V_+'}| + |E_{C_-, V_+'}| + |E_{V_-', V_+'}|
      \]
      Combining Equations $\eqref{eq:mm}+\eqref{eq:cutright}$,
      we obtain the upper bound of~$1.5x$.
    \item The cut $C_-',C_+'$ is ``after'' $C_-, C_+$, i.e., we have $C_+'
      \subseteq C_+$ and $C_-' \supseteq C_-$. In this case, we have $E_{C_+', V_+'} \subseteq E_{C_+,
      V_+'}$, and we always have $E_{V_-', C_-'} \subseteq E_{V_-', C_-} \sqcup
      E_{V_-',C_+}$. Thus:
      \[
        |E_{V_-', C_-'}| + |E_{C_+', V_+'}| + |E_{V_-', V_+'}| \leq 
        |E_{V_-', C_-}| + |E_{V_-',C_+}| + |E_{C_+,V_+'}| + |E_{V_-', V_+'}|
      \]
      Combining Equations $\eqref{eq:cutleft}+\eqref{eq:pp}$,
      we obtain the upper bound of~$1.5x$.
  \end{itemize}
  Hence, the second case of Claim~\ref{clm:choice} then holds. This concludes the
  proof.
  \end{proof}

\begin{figure}
\centering
\begin{tikzpicture}[
	ktext/.style={right,font=\footnotesize},
	xaxis/.style={black},
	req/.style={ultra thick},
	arwext/.style={-,blue!80},
	arwin/.style={-,orange},
	arw1/.style={-,red!80,thick},
	arw2/.style={-,magenta!80,very thick},
	arw3/.style={-,DarkGreen!80,densely dotted,thick},
	arw4/.style={-,Chartreuse!80!black,dotted,ultra thick},
	every label/.style={inner sep=0.05cm},
	xscale=1
]

\draw[brown!50,fill=brown!20] (5/2,0) ellipse (1cm and .2cm) node[brown,above right=.1cm and .45cm] {$C$};
\draw[xaxis] (0,0) -- +(5,0);
\draw[dashed] (5/2,-.5) --  +(0,2) node[above] (nc) {cut};
\node[below left,font=\small,inner sep=.2cm] at (nc) {$V_-'$~~~~~~~~~~~};
\node[below right,font=\small,inner sep=.2cm] at (nc) {~~~~~~~~~~~$V_+'$};
\draw[dotted,Grey!80!black] (0+1,-.5) --  +(0,2); %
\draw[dotted,Grey!80!black] (5-1,-.5) --  +(0,2); %

\draw[arwext] (0,.1) edge[bend left=30] (5,.1);
\draw[arwin] (5/2-.8,-.04) edge[bend right=40] (5/2+.8,-.04);
\draw[arw1] (0+.5,.1) edge[bend left=30] (5/2+.5,.1);
\draw[arw2] (5-.5,-.1) edge[bend left=30] (5/2-.5,-.1);
\draw[arw3] (0+.5,.1) edge[bend left=30] (5/2-.5,.1);
\draw[arw4] (5-.5,-.1) edge[bend left=30] (5/2+.5,-.1);

\begin{scope}[xshift=6cm,yshift=1.75cm,yscale=1.3]
\draw[Grey!30,thin] (-.1,.3) rectangle +(2.2,-2.1);
\draw[arwext] (0,0) -- +(.5,0) node[ktext] {$E_{V_-', V_+'}$};
\draw[arwin] (0,-.3) -- +(.5,0) node[ktext] {$E_{C, C}$};
\draw[arw1] (0,-.6) -- +(.5,0) node[ktext] {$E_{V_-, C_+}$};
\draw[arw2] (0,-.9) -- +(.5,0) node[ktext] {$E_{C_-, V_+}$};
\draw[arw3] (0,-1.2) -- +(.5,0) node[ktext] {$E_{V_-, C_-}$};
\draw[arw4] (0,-1.5) -- +(.5,0) node[ktext] {$E_{C_+, V_+}$};

\end{scope}
\end{tikzpicture}
\caption{Partitioning edges that may contribute to the cuts before or after reversing $<_C$.}\label{fig:cw-quotient-edges}
\end{figure}
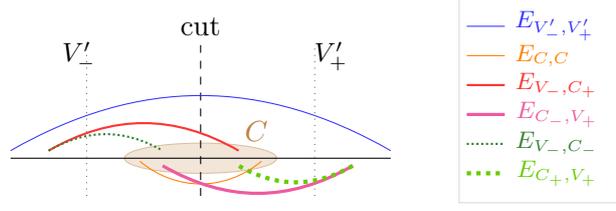

\section{Lower bound}
\label{sec:lower}

In this section, we show that the constant factor of~$1.5$ in the statement of 
Theorem~\ref{thm:upper} cannot be improved: we 
show in particular that the upper bound of
$x+y$ claimed by~\cite{FeldmannM23} does not hold. Here is the formal statement
that we show:

\begin{proposition}\label{prp:lower-bound-condensation}
For each even integer $x \geq 2$, for each integer $y \geq 1.5x$,
  there is a directed graph $H$ whose cutwidth is $1.5x + y$,
  where $x$ is the cutwidth of the condensation multigraph of~$H$,
  and $y$ is the maximum cutwidth of an SCC from $H$.
\end{proposition}

We show in Figure~\ref{fig:counterex} a directed graph obtained in our proof of
Proposition~\ref{prp:lower-bound-condensation} for the case where $x=2$ and $y=3$.
This directed graph has cutwidth 6 even though $x+y=5$.

To prove Proposition~\ref{prp:lower-bound-condensation}, we introduce
the undirected multigraph $G$ illustrated in
Figure~\ref{fig:cw-condensation-lower}. This multigraph has
vertex set $\{1,2,3,4,5\}$. It has edge $\{1,3\}$ with multiplicity $x$,
edges $\{2,5\}$
and $\{4,5\}$ each having multiplicity $x/2$, and edges $\{2,3\}$ and $\{3,4\}$
each having multiplicity $y$.
Further, we consider the vertex partitioning
$P = \{\{1\}$, $\{5\}$, $\{2,3,4\}\}$.
It is easy to see that the condensation multigraph $G/P$ has cutwidth $x$, and
that the multigraph induced by the class $\{2,3,4\}$ has cutwidth $y$ while the
multigraphs induced by the two other classes of~$P$ are trivial. It remains to
show that:

\begin{claim}
  \label{clm:cwlow}
The cutwidth of $G$ is $\min(1.5x + y, \max(2y, x))$.
\end{claim}

Let us first explain why the claim suffices to prove
Proposition~\ref{prp:lower-bound-condensation}. Let $x \geq 2$ be even and let $y \geq
1.5x$, and let us build the multigraph $G$ for these values of~$x$ and~$y$: by
Claim~\ref{clm:cwlow}, we know that $G$ has cutwidth $1.5x + y$, because $y \geq 1.5x$ ensures that
$2y \geq x$ and also ensures that $2y \geq 1.5x + y$, so the expression of the
claim evaluates to~$1.5x+y$. Now, let $G'$ be the undirected graph obtained
from~$G$ by subdividing each occurrence of each edge:
by repeated application of Proposition~\ref{prp:subdiv},
we know that $G'$ has the same cutwidth
as~$G$,
and $G'$ is now a simple undirected graph. Last, let
$H$ be the directed graph obtained by orienting the edges of~$G'$ so that the
only non-trivial connected component is the one featuring the vertices
$\{2,3,4\}$ along with the intermediate vertices added in the subdivision: as
$y\geq 1$, we can achieve this by picking different orientations for two of
the paths of length~$2$ which connect $2$ and $3$ in the subdivision, doing the
same for the paths of length~$2$ connecting $3$ and $4$, and ordering all other edges along a
topological sort to ensure that no other non-trivial SCCs exist.
Figure~\ref{fig:counterex} shows the graph $H$ obtained at the end of the
process (in the case where $x=2$ and $y=3$).

The end result
is a directed graph~$H$ whose condensation multigraph has the same cutwidth as the
quotient multigraph of~$G$ by~$P$ (hence, $x$) and where the cutwidth of SCCs is at
most $y$, but the cutwidth of~$H$ is the same as that of~$G$, i.e., it is
$1.5x+y$. This suffices to show Proposition~\ref{prp:lower-bound-condensation}.

\begin{figure}
\centering
\begin{tikzpicture}[
	ve/.style={minimum width=1.5mm,inner sep=0cm,draw,circle,Grey},
	cscc/.style={colscc!80!black},
	vscc/.style={minimum width=1.5mm,inner sep=0cm,draw,circle,colscc,inner color=colscc!50, outer color=colscc},
	mult/.style={inner sep=.5mm},
	arwblack/.style={black!80,>=stealth},
	arwscc/.style={colscc!80!black,>=stealth},
	every label/.style={inner sep=0.05cm,font=\small},
	every node/.style={font=\small},
]
\begin{scope}[
    xscale=1,
	yscale=.6
 ]
\node[ve,label={[label distance=.02cm,font=\footnotesize]180:$1$}] (p1) at (0,0) {};
\node[ve,label={[label distance=.02cm,font=\footnotesize]0:$5$}] (p2) at (6,0) {};
\node[vscc,label={[label distance=.02cm,font=\footnotesize]90:~}] (s3) at
  (3,.5) {};
\node[vscc,label={[label distance=.02cm,font=\footnotesize]90:~}] (s4) at
  (3,1.5) {};
\node[vscc,label={[label distance=.02cm,font=\footnotesize]90:~}] (s5) at
  (3,2.5) {};
\node[vscc,label={[cscc,label distance=.02cm,font=\footnotesize]45:$2$}] (p3) at (4,1.25) {};
\node[ve,label={[label distance=.02cm,font=\footnotesize]180:~}] (s1) at (1,-1) {};
\node[ve,label={[label distance=.02cm,font=\footnotesize]180:~}] (s2) at (1,1) {};
\node[vscc,label={[cscc,label distance=.05cm,font=\footnotesize]95:$3$}] (p4) at (2,0) {};
\node[vscc,label={[label distance=.02cm,font=\footnotesize]90:~}] (s6) at
  (3,-.5) {};
\node[vscc,label={[label distance=.02cm,font=\footnotesize]90:~}] (s7) at
  (3,-1.5) {};
\node[vscc,label={[label distance=.02cm,font=\footnotesize]90:~}] (s8) at
  (3,-2.5) {};
\node[vscc,label={[cscc,label distance=.02cm,font=\footnotesize]315:$4$}] (p5) at (4,-1.25) {};

\draw[arwblack,->] (p1) edge (s1);
\draw[arwblack,->] (p1) edge (s2);
\draw[arwblack,->] (s1) edge (p4);
\draw[arwblack,->] (s2) edge (p4);
\draw[arwblack,->] (p3) edge (p2);
\draw[arwblack,->] (p5) edge (p2);
\draw[arwscc,->] (p4) edge (s3);
\draw[arwscc,->] (p4) edge (s4);
\draw[arwscc,<-] (p4) edge (s5);
\draw[arwscc,->] (s3) edge (p3);
\draw[arwscc,->] (s4) edge (p3);
\draw[arwscc,<-] (s5) edge (p3);
\draw[arwscc,->] (p4) edge (s6);
\draw[arwscc,->] (p4) edge (s7);
\draw[arwscc,<-] (p4) edge (s8);
\draw[arwscc,->] (s6) edge (p5);
\draw[arwscc,->] (s7) edge (p5);
\draw[arwscc,<-] (s8) edge (p5);

\draw[dotted,brown,thin] (1.6,-2.8) rectangle (4.4,3.3);
\end{scope}
\begin{scope}[
	xshift=8cm,
	yshift=0cm,
    scale=.6
]
\node[ve,label={[label distance=.02cm,font=\footnotesize]180:$1$}] (p1) at (0,0) {};
\node[ve,label={[label distance=.02cm,font=\footnotesize]0:$5$}] (p2) at (6,0) {};
\node[ve,label={[label distance=.02cm,font=\footnotesize]180:~}] (s1) at (1.5,-1) {};
\node[ve,label={[label distance=.02cm,font=\footnotesize]180:~}] (s2) at (1.5,1) {};
\node[vscc,label={[label distance=.02cm,font=\footnotesize]225:~}] (p) at (3,0) {};

\draw[arwblack,->] (p1) edge (s1);
\draw[arwblack,->] (p1) edge (s2);
\draw[arwblack,->] (s1) edge (p);
\draw[arwblack,->] (s2) edge (p);
\draw[arwblack,->] (p) edge node[above=.5mm,mult] {$2$}(p2);
\draw[arwblack,->] (p) edge (p2);
\end{scope}
  \end{tikzpicture}
  \caption{Example of the directed graph $H$ obtained in the proof of
  Proposition~\ref{prp:lower-bound-condensation} (left), together with its
  condensation multigraph (right). For this graph, the maximum
  cutwidth of an SCC is $y=3$ (indeed the only nontrivial SCC is formed by the filled brown nodes ($\{2,3,4\}$ and the nodes connecting them), connected by dashed  brown edges), and the cutwidth of the condensation multigraph is
  $x = 2$, however it can be checked that the cutwidth of~$H$ is $1.5x + y = 6$ and not $x + y = 5$.}
  \label{fig:counterex}
\end{figure}
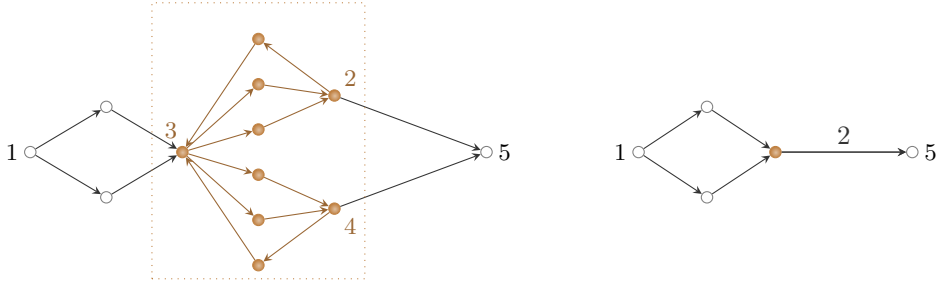

Hence, all that remains is to show the claim:

  \begin{proof}[Proof of Claim~\ref{clm:cwlow}]
First, we simplify the graph. 
To simplify the proof, we observe that by \cref{prp:subdiv}, $G$ has the same cutwidth as the multigraph $K$ obtained by removing node $5$ and replacing its incident edges by $x/2$ edges connecting $2$ to $4$. 
Consequently, we want to show that $K$ has the cutwidth stated by the claim.

In $K$, ordering $2,3,1,4$ achieves a cutwidth of $1.5x + y$, 
whereas the number of edges that cross the successive cuts of ordering $2,4,3,1$ are
    respectively $y+x/2$, $2y$, and $x$ so that this ordering has cutwidth $\max(2y, x)$.
    (It cannot be the case that the maximum is achieved by $y+x/2$ because this
    would imply $y+x/2 > 2y$ so that $x/2 > y$ but then $y+x/2 < x$.)
    This proves that the cutwidth of $G$ is at most $\min(1.5x + y, \max(2y,
    x))$.
    (Note that the upper bound of $1.5x+y$ could alternatively be viewed as a
    consequence of Theorem~\ref{thm:upper}.)

    Let us now prove a corresponding lower bound on the cutwidth of~$K$.
Let $C$ denote the subset $\{2,3,4\}$ of vertices of~$K$, and let $<_K$ denote an ordering on the $4$ nodes of $K$. 
Consider first the case where the ordering is monotone on the identifiers of
    nodes from~$C$. Up to reversing the ordering (which does not change the
    cutwidth), we may assume $2 <_K 3 <_K 4$. If $1 <_K 3$ then  
$K$ contains $1.5x + y$ edges crossing the cut $(\{2,1\}, \{3,4\})$. If $1
    >_K 3$, then $K$ contains $1.5x + y$ edges crossing the cut $(\{2,3\}, \{1,4\})$. This shows that such monotone orders have cutwidth at least $1.5x + y$.

Consider now the case where the ordering $<_K$ is non-monotone. Up to reversing
    the order, and because 2 and 4 play symmetric roles, we can assume that
    the ordering $<_K$ is $2 <_K 4 <_K 3$. Now, any cut of $K$ separating $\{2,4\}$ from $\{3\}$ is crossed by $2y$ edges, whereas any cut separating $1$ from $3$ is crossed by $x$ edges.
This shows that such non-monotone orders have cutwidth at least $\max(2y, x)$.

As a consequence, $K$ and $G$ have cutwidth $\min(1.5x + y, \max(2y, x))$, which
    concludes the proof of the claim.
\end{proof}
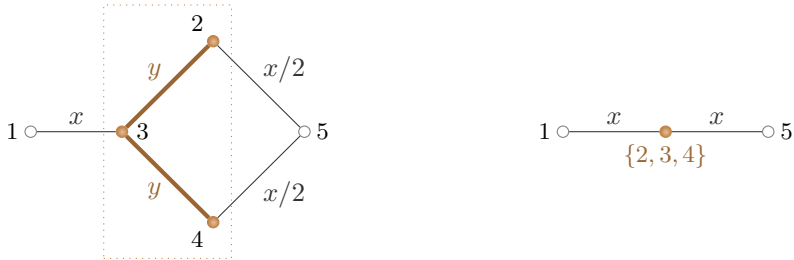
\begin{figure}
\centering
\begin{tikzpicture}[
	ve/.style={minimum width=1.5mm,inner sep=0cm,draw,circle,Grey},
	cscc/.style={colscc!80!black},
	vscc/.style={minimum width=1.5mm,inner sep=0cm,draw,circle,colscc,inner color=colscc!50, outer color=colscc},
	mult/.style={inner sep=.5mm},
	arwblack/.style={black!80,>=stealth},
	arwscc/.style={colscc!80!black,>=stealth, ultra thick},
	every label/.style={inner sep=0.05cm,font=\small},
	every node/.style={font=\small},
]
\begin{scope}[
	scale=.6
 ]
\node[ve,label={[label distance=.02cm,font=\footnotesize]180:$1$}] (p1) at (0,0) {};
\node[ve,label={[label distance=.02cm,font=\footnotesize]0:$5$}] (p2) at (6,0) {};
\node[vscc,label={[label distance=.02cm,font=\footnotesize]135:$2$}] (p3) at (4,2) {};
\node[vscc,label={[label distance=.05cm,font=\footnotesize]0:$3$}] (p4) at (2,0) {};
\node[vscc,label={[label distance=.02cm,font=\footnotesize]225:$4$}] (p5) at (4,-2) {};

\draw[arwblack] (p1) edge node[above=.5mm,mult] {$x$}(p4);
\draw[arwblack] (p3) edge node[above right,mult] {$x/2$}(p2);
\draw[arwblack] (p5) edge node[below right,mult] {$x/2$}(p2);
\draw[arwscc] (p3) edge node[above left,mult] {$y$}(p4);
\draw[arwscc] (p4) edge node[below left,mult] {$y$}(p5);

\draw[dotted,brown,thin] (1.6,-2.8) rectangle (4.4,2.8);
\end{scope}

\begin{scope}[
	xshift=7cm,
	yshift=0cm,
    scale=.45
]
\node[ve,label={[label distance=.02cm,font=\footnotesize]180:$1$}] (p1) at (0,0) {};
\node[ve,label={[label distance=.02cm,font=\footnotesize]0:$5$}] (p2) at (6,0) {};
\node[vscc,label={[colscc!80!black,label distance=.02cm,font=\footnotesize]270:$\{2,3,4\}$}] (p) at (3,0) {};

\draw[arwblack] (p1) edge node[above=.5mm,mult] {$x$} (p);
\draw[arwblack] (p) edge node[above=.5mm,mult] {$x$} (p2);
\end{scope}

\end{tikzpicture}
  \caption{Multigraph $G$ used in the proof of
  Proposition~\ref{prp:lower-bound-condensation} (left) and its condensation multigraph (right).} \label{fig:cw-condensation-lower}
\end{figure}

Let us last remark that
the lower bound shown in Proposition~\ref{prp:lower-bound-condensation} holds for the
specific graphs that we construct, but not for general graphs. More precisely,
consider a directed graph $H$ and its condensation multigraph $H'$. The cutwidth of~$H$
cannot be less than the maximal cutwidth
$y$ of an SCC of~$G$, because SCCs are subgraphs of~$G$ and cutwidth is monotone
under taking subgraphs. However, the cutwidth $x$ of~$H'$ is not a lower bound
on the cutwidth of~$G$. Indeed:

\begin{proposition}
  \label{prp:nolow}
  For each $n \geq 3$, we can build a directed graph $G_n$ where the
SCC condensation multigraph has cutwidth~$x = n$, the maximal cutwidth of an SCC
is~$2$, but the graph $G_n$ itself has constant cutwidth.
\end{proposition}

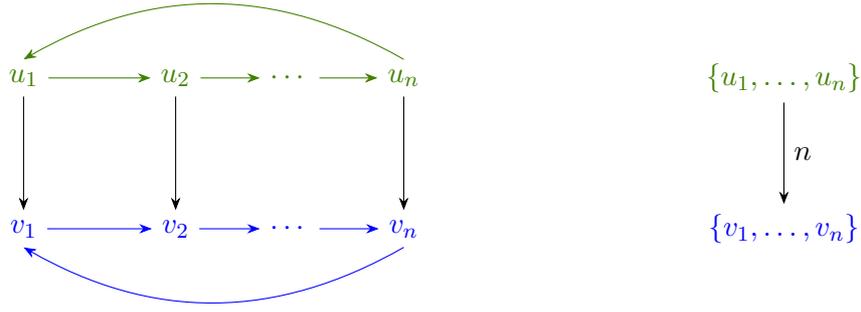
\begin{figure}
\centering
\begin{tikzpicture}[->,>=Stealth, node distance=1.5cm and 2cm, auto,
scc1/.style={Chartreuse!50!black},
scc2/.style={blue},
]
\begin{scope}
  \node[scc1] (u1) at (0, 0) {$u_1$};
  \node[scc1] (u2) at (2, 0) {$u_2$};
  \node[scc1] (udots) at (3.5, 0) {$\cdots$};
  \node[scc1] (un) at (5, 0) {$u_n$};

  \node[scc2] (v1) at (0, -2) {$v_1$};
  \node[scc2] (v2) at (2, -2) {$v_2$};
  \node[scc2] (vdots) at (3.5, -2) {$\cdots$};
  \node[scc2] (vn) at (5, -2) {$v_n$};

  \draw[scc1] (u1) -- (u2);
  \draw[scc1] (u2) -- (udots);
  \draw[scc1] (udots) -- (un);
  \draw[scc1,bend right] (un.north) to (u1.north);

  \draw[scc2] (v1) -- (v2);
  \draw[scc2] (v2) -- (vdots);
  \draw[scc2] (vdots) -- (vn);
  \draw[scc2,bend left] (vn.south) to (v1.south);

  \draw (u1) -- (v1);
  \draw (u2) -- (v2);
  \draw (un) -- (vn);
\end{scope}

\begin{scope}[
	xshift=10cm
]
    \node[scc1] (a) at (0, 0) {$\{u_1, \ldots, u_n\}$};
    \node[scc2] (b) at (0, -2) {$\{v_1, \ldots, v_n\}$};
    \draw[->] (a) edge node[right] {$n$} (b);
\end{scope}
  \end{tikzpicture}
  \caption{Graph $G_n$ used in the proof of Proposition~\ref{prp:nolow} (left),
  and its condensation multigraph (right).}
  \label{fig:nolow}
\end{figure}

\begin{proof}
  We picture the graph $G_n$ used to prove the result in Figure~\ref{fig:nolow}
  (left).

  The graph $G_n = (V_n, E_n)$ has vertices $V_n = \{u_1, \ldots, u_n\} \sqcup
  \{v_1, \ldots, v_n\}$ and has as edges $E_n$ two directed cycles $(u_1, u_2), \ldots,
  (u_{n-1}, u_n), (u_n, u_1)$ and $(v_1, v_2), \ldots,
  (v_{n-1}, v_n), (v_n, v_1)$, along with $n$ edges $(u_i, v_i)$ for each $1
  \leq i \leq n$. Thus, the SCCs of~$G_n$ are $\{u_1, \ldots, u_n\}$ and 
  $\{v_1, \ldots, v_n\}$, each of which is a cycle with cutwidth~$2$.

  The
  condensation multigraph 
  multigraph (pictured in Figure~\ref{fig:nolow}, right)
  consists of two vertices connected by $n$ parallel
  edges, so it has cutwidth~$n$. However, the ordering $u_1 < v_1 < \cdots < u_n
  < v_n$ achieves a cutwidth of at most 5: in each cut we have at most 2
  internal edges cut in each of the two SCCs, and at most one external edge cut.
  This concludes the proof.
\end{proof}

\section{Conclusion and Future Work}
\label{sec:conc}

We have shown that, for any graph $G$ and vertex partitioning of $G$, if we let
$x$ be the cutwidth of the quotient multigraph of~$G$ and let $y$ be the maximal
cutwidth of a class of~$G$, then the cutwidth of~$G$ can be bounded with
$1.5x+y$; and there are arbitrarily large values of $x$ and $y$ for which this
bound cannot be improved.

Our work also leaves open the question of whether better bounds can be shown, in
particular whether the upper bound of \cref{thm:upper} can be improved in some
settings. In particular, 
our lower bound of \cref{prp:lower-bound-condensation} is only shown in the case
where $y \geq 1.5x$, leaving open the question of whether better upper bounds
than \cref{thm:upper} can be shown in the setting where $y < 1.5x$.
Another more general question for future research is whether the vertex
partitioning approach that we studied could also be applied to other width measures beyond cutwidth.

\bibliographystyle{apalike}
\bibliography{references}

\vfill
\doclicenseThis

\end{document}